\documentclass[11pt,a4paper]{article}
\usepackage[utf8]{inputenc}

\usepackage{amsmath}
\usepackage{amsfonts}
\usepackage{amssymb}
\usepackage{amsthm}
\usepackage{graphicx}
\usepackage{lmodern}
\usepackage[margin=3.2cm]{geometry}
\usepackage{hyperref}
\usepackage[dvipsnames]{xcolor}
\usepackage{authblk}
\usepackage{mathrsfs}
\usepackage[all]{nowidow}
\usepackage{microtype}

\newcommand{\llangle}{\langle \!\langle}
\newcommand{\rrangle}{\rangle \!\rangle}
\newcommand{\fps}[2]{#1 \llangle #2 \rrangle}
\newcommand{\pol}[2]{#1 \langle #2\rangle}
\newcommand{\Ln}{\mathscr{L}}

\newcommand{\Gu}{\mathscr{G}}

\newcommand{\cayG}{\mathscr{C}} 
\newcommand{\G}[1]{\normalfont{\texttt{[}}\mathord{#1}\normalfont{\texttt{]}}}
\newcommand{\Gone}{\G{\mathtt{1}}}
\newcommand{\Gzero}{\G{\mathtt{0}}}
\newcommand{\Gdia}{\bullet}
\newcommand{\Gdot}{\mathbin{\G{\cdot}}}
\newcommand{\Gplus}{\mathbin{\G{+}}}
\newcommand{\Gstar}{\mathbin{\G{\ast}}}

\theoremstyle{plain}
	\newtheorem{theorem}{Theorem}
	\newtheorem{definition}{Definition}
	\newtheorem{lemma}{Lemma}
	
	\newtheorem{corollary}{Corollary}
	\newtheorem*{claim}{Claim}
\theoremstyle{definition}
	\newtheorem{example}{Example}

\title{Completeness of Finitely Weighted Kleene Algebra With Tests\thanks{
		Published version: I. Sedlár: Completeness of Finitely Weighted Kleene Algebra with Tests. In: G. Metcalfe, T. Studer, R. de Queiroz (Eds.): \emph{Logic, Language, Information, and Computation (WoLLIC 2024)}, pp.~210--224. Lecture Notes in Computer Science, vol 14672. Springer, Cham, 2024. 	DOI: \href{https://doi.org/10.1007/978-3-031-62687-6_14}{10.1007/978-3-031-62687-6\_14}
	}
}
\author{Igor Sedlár}
\affil{The Czech Academy of Sciences, Institute of Computer Science \authorcr
	Prague, The Czech Republic}
\date{ }	

\begin{document}
\maketitle

\vspace*{-1cm}
\begin{abstract}
Building on Ésik and Kuich's completeness result for finitely weighted Kleene algebra, we establish relational and language completeness results for finitely weighted Kleene algebra with tests. Similarly as Ésik and Kuich, we assume that the finite semiring of weights is commutative, partially ordered and zero-bounded, but we also assume that it is integral. We argue that finitely weighted Kleene algebra with tests is a natural framework for  equational reasoning about weighted programs in cases where an upper bound on admissible weights is assumed.
\end{abstract}

\section{Introduction}

Ésik and Kuich \cite{EsikKuich2001} generalize the completeness result of Kozen \cite{Kozen1994}, which connects Kleene algebras with the algebra of regular languages, to the case of weighted regular languages, or formal power series. In particular, their result applies to a weighted generalization of Kleene algebra where the semiring of weights is finite, commutative, zero-bounded (or positive) and partially ordered. Building on their work, we establish two completeness results for a weighted generalization of Kleene algebras with tests \cite{Kozen1997}. First, we establish completeness with respect to the algebra of weighted guarded languages using a reduction to weighted regular languages similar to the one used by Kozen and Smith \cite{KozenSmith1997} to prove completeness for non-weighted Kleene algebras with tests. Second, we establish completeness with respect to weighted transition systems by using a Cayley-like construction going back to the work of Pratt \cite{Pratt1980} on (non-weighted) dynamic algebras. In addition to the assumptions of Ésik and Kuich, we need to assume that the semiring of weights is also integral. 

Weighted Kleene algebras with tests provide a framework for equational reasoning about properties of weighted programs \cite{BatzEtAl2022}, a generalization of standard programs articulating the idea that computation may carry some sort of weight. We have pointed out this connection in earlier work \cite{Sedlar2023a}; this paper contributes with a slightly simpler but more general definition of a weighted Kleene algebra with tests, and the completeness results. We also argue that although \emph{finitely} weighted structures are a simplification with respect to the cases usually studied when weighted computation is concerned \cite{DrosteEtAl2009,KuichSalomaa1986,SalomaaSoittola1978} they are still practically relevant since they represent the assumption of an upper bound on admissible weights. 


The paper is organized as follows. Section \ref{sec:semirings} recalls semirings, their interpretation in terms of weights, and discusses formal power series and weighted relations. Section \ref{sec:WeKAT} introduces weighted Kleene algebras with tests and outlines their applications in reasoning about weighted programs. Completeness of finitely weighted Kleene algebras with tests with respect to the algebra of weighted guarded languages is established in Section \ref{sec:Lan} and Section \ref{sec:Rel} establishes completeness with respect to weighted transition systems. The concluding Section \ref{sec:conclusion} summarizes the paper and outlines some interesting problems we leave for future work. Details of some of the proofs are given in the technical appendix.

\section{Semirings and weighted structures}\label{sec:semirings}

In this section we recall the basic background information on semirings, and we discuss their interpretation as algebras of weights. Then we recall two kinds of weighted structures that will be important in this paper: formal power series (weighted languages) and weighted relations (square matrices over a semiring). Our discussion is based on \cite{KuichSalomaa1986}, and details are provided for the benefit of the non-specialist reader. At the end of the section, we discuss the usefulness of working with finite semirings of weights.

Recall that a \emph{semiring} is a structure $\langle X, +, \cdot, 0, 1\rangle$ where $\langle X, \cdot, 1\rangle$ is a monoid, $\langle X, +, 0\rangle$ is a commutative monoid, multiplication $\cdot$ distributes over addition $+$ from both sides, and $0$ is the multiplicative annihilator ($0 \cdot x = 0 = x \cdot 0$ for all $x \in X$). A semiring $X$ is \emph{commutative} if $x \cdot y = y \cdot x$ for all $x,y \in X$; it is (additively) \emph{idempotent} iff $x + x = x$ for all $x \in X$; and $X$ is \emph{partially ordered} if there is a partial order $\leq$ on $X$ such that both $\cdot$ and $+$ are monotone with respect to $\leq$. A partially ordered semiring $X$ is \emph{zero-bounded} if $0 \leq x$ for all $x \in X$, and it is \emph{integral} if $x \leq 1$ for all $x \in X$. Note that each idempotent semiring is partially ordered ($x \leq y$ iff $x + y = y$) and zero-bounded. A commutative partially ordered zero-bounded semiring is called a \emph{copo-semiring}; a copo-semiring which is also integral is called a \emph{copi-semiring}.

Semirings are natural models of \emph{weights} (for instance, the amount of some resource needed to perform an action). Multiplication represents merge of weights ($x \cdot y$ as weight $x$ together with weight $y$), the multiplicative unit $1$ represents inert weight (resp.~no weight) and the annihilator $0$ represents absolute weight. Partially ordered semirings represent the idea that weights are ordered, but the order may not be linear. In zero-bounded ordered semirings $0$ is the ``worst'' weight and in integral semirings $1$ is the ``best'' one. Idempotent semirings come with a natural interpretation of addition, which can be seen as selecting the ``best'' weight out of a pair of weights. 

\begin{example}
A familiar example of a finite semiring is the \emph{Boolean semiring} $\langle \{ 0, 1 \}, \lor, \land, 0, 1\rangle$, representing an ``all or nothing'' perspective on weights. Real numbers $\mathbb{R}$ with their usual addition and multiplication also form a semiring, as do their subalgebras such as integers $\mathbb{Z}$ and natural numbers $\mathbb{N}$. A semiring often used in connection with shortest paths is the \emph{tropical semiring} over extended natural numbers $\mathbb{N} \cup \{ \infty \}$ where semiring multiplication is addition on $\mathbb{N} \cup \{ \infty \}$ (where $\infty + x = \infty = x + \infty$) and semiring addition is the minimum function, which is idempotent. This means that the  multiplicative identity $0$ (the natural number zero) is the greatest element with respect to the partial order induced by $\mathrm{min}$ and $\infty$, where the latter is the annihilator and the least element. A semiring often used in modelling imprecise or vague notions is the \emph{{\L}ukasiewicz semiring} on the real unit interval $[0,1]$ where $\max$ is semiring addition (idempotent) and the semiring multiplication is the {\L}ukasiewicz t-norm: $x \otimes y = \max\{ 0, x+y-1 \}$. Of course, $1$ is the multiplicative identity and $0$ is the annihilator.  
\end{example}    

Recall that if $\Sigma$ is a set, then $\Sigma^{*}$ is the set of all finite sequences of elements of $\Sigma$, including the empty sequence $\epsilon$. If $\Sigma$ is seen as an alphabet, $\Sigma^{*}$ can be seen as the set of all finite words over $\Sigma$. Algebraically speaking, $\Sigma^{*}$ is the free monoid generated by $\Sigma$ with $\epsilon$ as the multiplicative identity and word concatenation as multiplication. A \emph{language} over $\Sigma$ is a subset of $\Sigma^{*}$.
 A \emph{formal power series} over $\Sigma$ with coefficients in a semiring $S$ is a function from $\Sigma^{*}$ to $S$. Formal power series can be seen as \emph{weighted languages} and languages over $\Sigma$ correspond to formal power series over $\Sigma$ with coefficients in the Boolean semiring. The set of all formal power series over $\Sigma$ with coefficients in $S$ is usually denoted as $\fps{S}{\Sigma^{*}}$. A \emph{polynomial} is a formal power series $r$ such that the set of $w \in \Sigma^{*}$ with $r(w) \neq 0$ is finite. The set of polynomials in $\fps{S}{\Sigma^{*}}$ is usually denoted as $\pol{S}{\Sigma^{*}}$. Examples of polynomials in $\pol{S}{\Sigma^{*}}$ include formal power series usually denoted as $sw$, for $s \in S$ and $w \in \Sigma^{*}$, such that $sw(u) = s$ if $w = u$ and $sw(u) = 0$ otherwise. It is customary to denote $s\epsilon$ as $s$ and $1w$ as $w$ for $w \neq \epsilon$. The set $\fps{S}{\Sigma^{*}}$ has a structure of a semiring where $0, 1 \in \fps{S}{\Sigma^{*}}$ are polynomials of the type just mentioned, and
 \[ (r_1 + r_2)(w) = r_1(w) + r_2(w) \qquad
 (r_1 \cdot r_2)(w) = \sum_{w = w_1w_2} r_1(w_1) \cdot r_2(w_2) \, .\]
 (Note that the sum exists in arbitrary $S$ since there are only finitely many $w_1, w_2$ such that $w = w_1w_2$.) If $S$ is finite, then we may define $s^{*} = \sum_{n \in \omega} s^{n}$, where $s^{0} = 1$ and $s^{n+1} = s^{n} \cdot s$.\footnote{Not only if $S$ is finite, of course, but we need not consider this more general setting in this paper. See \cite{DrosteEtAl2009,KuichSalomaa1986,SalomaaSoittola1978}.} The $^*$ may be lifted to formal power series:
 \[ r^{*}(w) = \sum_{w_1 \ldots w_n \in [w]} r(w_1) \cdot \ldots \cdot r(w_n) \, ,\]
where $[w]$ is the set of all factors of $w$. (Equivalently, $r^{*}(w) = \sum_{n \in \omega} r^{n}(w)$.) Addition, multiplication and star are known as \emph{rational operations} on formal power series. The smallest set of formal power series in $\fps{S}{\Sigma^{*}}$ that contains all polynomials and is closed under the rational operations is known as the set of \emph{rational power series} and denoted as $\fps{S^{\mathrm{rat}}}{\Sigma^{*}}$. Rational power series over $\Sigma$ with coefficients in the Boolean semiring are the \emph{regular languages} over $\Sigma$.
 
 Let $S$ be a semiring. An \emph{$S$-weighted relation} on a set $Q$ is a function from $Q \times Q$ to $S$. An $S$-weighted relation on $Q$ may be seen as a (potentially infinite) $Q \times Q$ matrix with entries in $S$. Matrix addition can then be used to define weighted union of relations on $Q$ and, if $S$ is finite, matrix multiplication can be used to define weighted relational composition of relations on $Q$: 
 \[ (M+N)_{q, q'} = M_{q,q'} + N_{q,q'} \qquad
 (M \cdot N)_{q,q'} = \sum_{p \in Q} M_{q,p} \cdot N_{p,q'} \, .
 \]
 The set of all $S$-weighted relations on $Q$ ($Q \times Q$ matrices with entries in $S$) forms a semiring, where $0$ is the zero matrix (all entries are $0$) and $1$ is the identity matrix (the diagonal matrix where all entries on the diagonal are $1$). If $S$ is finite, then the star of a $Q \times Q$ matrix with entries in $S$ \cite{Conway1971,KuichSalomaa1986} represents weighted reflexive transitive closure of the corresponding $S$-weighted relation on $Q$: $M^{*} = \sum_{n \in \omega} M^{n}$, where $M^{0} = 1$ and $M^{n+1} = M^{n} \cdot M$. Alternatively, $M^{*}_{q,q'} = \sum_{x \in [q,q']} M(x)$ where $[q,q']$ is the set of all finite paths from $q$ to $q'$ and $M(q_1 \ldots q_n) = M_{q_1, q_2} \cdot \ldots \cdot M_{q_{n-1}, q_n}$. 
 
Weighted relations represent the idea that the transition from $q$ to $q'$ carries a weight, $M_{q,q'}$. For instance, if $\mathbb{N} \cup \{ \infty \}$ is used and $M$ gives each edge in a directed graph over $Q$ a weight of $1 \in \mathbb{N}$, then $M_{q,q'}$ is the length of the shortest path from $q$ to $q'$ -- literally the smallest number of steps in the graph you need to take in order to get from $q$ to $q'$ ($\infty$ if there is no path from $q$ to $q'$). Assuming an infinite semiring of weights, such as $\mathbb{N} \cup \{ \infty \}$, means that ``no path is too long \textit{a priori}''. On the other hand, in most practical situations it is sensible (or even necessary), to set a cut-off point beyond which all paths \emph{are} ``too long''. Reasoning about weights with a cut-off point in place corresponds to working with a \emph{finite} semiring of weights. 

\begin{example}
A \emph{finite tropical semiring} is any $N \cup \{ \infty \}$ where $N$ is a non-empty initial segment of $\mathbb{N}$, semiring multiplication is truncated addition ($n + m = \infty$ if the ``ordinary sum'' of $n$ and $m$ is not in $N$), semiring addition is $\mathrm{min}$, the multiplicative identity is $0$ and the annihilator is $\infty$. For $n > 1$, the \emph{$n$-element {\L}ukasiewicz semiring} is defined on the domain $\{ \frac{m}{n - 1} \mid m \in [0, n - 1] \}$ and the operations are defined as in the infinite case. Note that both of these examples yield \emph{nilpotent} semirings: for each $s \neq 1$ there is $n \in \omega$ such that $s^{n} = 0$.
\end{example}

\section{Weighted Kleene algebra with tests}\label{sec:WeKAT}

In this section, we recall Kleene algebras \cite{Kozen1994}, we formulate a notion of weighted Kleene algebra based on Ésik and Kuich's definition \cite{EsikKuich2001}, and we state their completeness result (Section \ref{sec:WeKAT-1}). Then we recall Kleene algebras with tests and we formulate a notion of weighted Kleene algebras with tests (\ref{sec:WeKAT-2}). Finally, applications to reasoning about weighted programs are outlined (\ref{sec:WeKAT-3}). 

\subsection{Kleene algebras}\label{sec:WeKAT-1}

Recall that a \emph{Kleene algebra} \cite{Kozen1994} is an idempotent semiring $X$ with a unary operation $^{*}$ satisfying, for all $x, y, z \in X$ the following \emph{unrolling} (left column) and \emph{fixpoint} laws (right column):
\begin{align}
1 + (x \cdot x^{*}) & = x^{*} &
y + (x \cdot z) \leq z & \implies x^{*} \cdot y \leq z\\
1 + (x^{*} \cdot x) & = x^{*} &
y + (z \cdot x) \leq z & \implies y \cdot x^{*} \leq z \, .
\end{align}
Ésik and Kuich \cite{EsikKuich2001} do not assume the unrolling law $1 + (x^{*} \cdot x) = x^{*}$ but it can be shown that it can be derived from the rest. We usually write $xy$ instead of $x \cdot y$. Some useful equalities that follow from these axioms are the \emph{sliding} and \emph{denesting} laws: $(xy)^{*}x = x(yx)^{*}$ and $(x + y)^{*} = x^{*}(yx^{*})^{*}$.

\begin{definition}
Let $S$ be a finite semiring. A \emph{Kleene $S$-algebra} is a Kleene algebra $X$ together with a binary operation \[ \odot : X \times S \to X \] such that (the additive monoid reduct of) $X$ forms a right $S$-semimodule and 
\begin{gather}
(x \cdot y) \odot s = x \cdot (y \odot s) = (x \odot s) \cdot y\label{e:s-comm}\\
1 \odot s^{*} \leq (1 \odot s)^{*}\label{e:s-star}
\end{gather}
\end{definition}
(Recall that $s^{*} = \sum_{n \in \omega} s^{n}$.) It follows that $(1 \odot s^{*}) = (1 \odot s)^{*}$; see \cite{EsikKuich2001}.\footnote{A similar definition could be formulated for the case of infinite $S$ where it is assumed that $s^{*}$ is defined for all $s$. For instance, one could assume that $S$ is a Kleene algebra.}

The class of all Kleene $S$-algebras for a fixed $S$ is denoted as $\mathsf{KA}(S)$; moreover, we define $\mathsf{KA}(\mathsf{X}) = \bigcup_{S \in \mathsf{X}} \mathsf{KA}(S)$. 
 Ésik and Kuich use the \emph{left} $S$-action on $K$ instead of the right one (that is, $\odot : S \times X \to X$).

The following syntactic definition is not used explicitly by Ésik and Kuich, but we will need it in this paper. (Ésik and Kuich state their result using the notion of free algebra.)

\begin{definition}
Let $\Sigma$ be a finite alphabet and $S$ a finite semiring. The set of $(\Sigma, S)$-expressions $\mathrm{Exp}(\Sigma, S)$ is defined using the following grammar:
\[ e, f := \mathtt{a} \mid e \odot s \mid e + f \mid e \cdot f \mid e^{*} \mid \mathtt{0} \mid \mathtt{1}\] where $\mathtt{a} \in \Sigma$ and $s \in S$.
\end{definition}

A Kleene $S$-algebra \emph{model} is a Kleene $S$-algebra $X$ together with a homomorphism $h : \mathrm{Exp}(\Sigma, S) \to X$, where it is understood that $h(e \odot s) = h(e) \odot s$. The notions of validity of an equation in a class of Kleene $S$-algebras, and equational theory of a class of Kleene $S$-algebras are defined in the usual way. Note that $\fps{S^{\mathrm{rat}}}{\Sigma^{*}} \in \mathsf{KA}(S)$.

\begin{theorem}[Ésik and Kuich \cite{EsikKuich2001}]\label{thm:EK}
If $S$ is a finite copo-semiring, then, for all $e,f \in \mathrm{Exp}(\Sigma, S)$, \[ \mathsf{KA}(S) \models e \approx f \iff \fps{S^{\mathrm{rat}}}{\Sigma^{*}} \models e \approx f \, .\]
\end{theorem}

The \emph{standard language interpretation} of $\mathrm{Exp}(\Sigma, S)$ is the unique homomorphism $\Ln : \mathrm{Exp}(\Sigma, S) \to \fps{S}{\Sigma^{*}}$ such that $\Ln (\mathtt{a}) = 1\mathtt{a}$.

\begin{corollary}\label{coro:EK}
If $S$ is a finite copo-semiring, then, for all $e,f \in \mathrm{Exp}(\Sigma, S)$, \[ \mathsf{KA}(S) \models e \approx f \iff \Ln (e) = \Ln (f) \, .\]
\end{corollary}
\begin{proof}
The non-trivial implication follows from the fact that for every $r \in \fps{S^{\mathrm{rat}}}{\Sigma^{*}}$ there is $e_{r}$ such that $r = \Ln (e_r)$.
\end{proof}

\subsection{Kleene algebras with tests}\label{sec:WeKAT-2}

Recall that a Kleene algebra with tests \cite{Kozen1997} is a Kleene algebra $X$ with a distinguished $B \subseteq X$ such that $\langle B, +, \cdot, 0, 1\rangle$ is a subalgebra of $X$ and a bounded distributive lattice, and $\,^{-}$ is an unary operation on $B$ such that $x \cdot \bar{x} = 0$ and $x + \bar{x} = 1$ for all $x \in B$. Hence, $B$ forms a Boolean algebra. Intuitively, elements of $B$ represent \emph{Boolean tests}. 

\begin{definition}
Let $S$ be a semiring. A \emph{Kleene $S$-algebra with tests} is a Kleene $S$-algebra $X$ that is also a Kleene algebra with tests.
\end{definition}
Recall that every Kleene algebra is also a Kleene algebra with tests (take $B = \{ 0, 1 \}$ with $\,^{-}$ defined in the obvious way). Hence, every Kleene $S$-algebra is a Kleene $S$-algebra with tests. Conversely, every Kleene $S$-algebra with tests becomes a Kleene $S$-algebra if one ``disregards'' the tests. The class of all Kleene $S$-algebras with tests for a fixed $S$ is denoted as $\mathsf{KAT}(S)$, and $\mathsf{KAT(X)}$ is defined in the obvious way. Concrete examples of Kleene $S$-algebras with tests will be discussed in the following two sections.

\begin{definition}
Let $S$ be a finite semiring and $\Sigma, \Phi$ be two finite mutually disjoint alphabets. The set of $\Phi$-tests $\mathrm{Te}(\Phi)$ is defined using the following grammar:
\[ b, c := \mathtt{p} \mid \bar{b} \mid b + c \mid b \cdot c \mid \mathtt{0} \mid \mathtt{1} \, ,\]
where $\mathtt{p} \in \Phi$. The set of $(\Sigma, \Phi, S)$-expressions $\mathrm{Exp}(\Sigma, \Phi, S)$ is defined using the following grammar:
\[ e, f := \mathtt{a} \mid b \mid  e \odot s \mid e + f \mid e\cdot f \mid e ^{*} \, ,\] where $\mathtt{a} \in \Sigma$, $s \in S$, and $b \in \mathrm{Te}(\Phi)$.
\end{definition}

A \emph{$\mathsf{KAT}(S)$-model} is a Kleene $S$-algebra with tests $X$ together with a homomorphism $h : \mathrm{Exp}(\Sigma, \Phi, S) \to X$ where $h(e \odot s) = h(e) \odot s$ and $h(\mathtt{p}) \in B$. The notion of validity in a class of Kleene $S$-algebras with tests and the equational theory of a class of Kleene $S$-algebras with tests are defined in the usual way.

\subsection{Weighted programs}\label{sec:WeKAT-3}

Batz et al.~\cite{BatzEtAl2022} consider an extension of the language of \emph{while programs} (containing variable assignment commands, the skip command, sequential composition, conditionals and while loops) with \emph{non-deterministic branching} and \emph{weighting}. For an element $s \in S$ of some fixed semiring of weights, the command ``add $s$'' merges the weight of the current computation path with $s$. Batz et al.~\cite{BatzEtAl2022} argue that these \emph{weighted programs} constitute a useful formalism that captures certain mathematical models (such as optimization problems) in an intuitive algorithmic way. They also argue that weighted programs generalize \emph{probabilistic programs}. For further details, the reader is referred to \cite{BatzEtAl2022}.

It is well known that Kleene algebras with tests are able to express the control flow commands of while programs: sequential composition $e;f$ is expressed as $e \cdot f$; \textbf{skip} is $1$; \textbf{if} $b$ \textbf{then} $e$ \textbf{else} $f$ is $be + \bar{b}f$; and \textbf{while} $b$ \textbf{do} $e$ is $(be)^{*}\bar{b}$.
 Non-deterministic branching corresponds directly to $e + f$. Importantly, the expression $e \odot s$ represents the sequential composition of $e$ with the weighting by $s$: ``do $e$ and then add weight $s$''. 

\begin{example}
As an example, we formalize the \emph{Ski Rental Problem} of \cite{BatzEtAl2022} as an expression of Kleene $S$-algebra with tests. The \textit{SRP} is based on the following scenario: ``A person does not own a pair of skis but is going on a skiing trip for $n$ days. At the beginning of each day, the person can chose between two options: Either rent a pair of skis, costing $1$ EUR for that day; or buy a pair of skis, costing $n$ EUR (and then go skiing for all subsequent days free of charge).'' \cite[p.~3]{BatzEtAl2022}. The semiring of weights in $\mathbb{N} \cup \{ \infty \}$. The weighted program representing the situation is:
\[ \textbf{while } \mathtt{n > 0} \textbf{ do }
\left( \mathtt{n := n-1} ; \left( \odot \textbf{1} + (\odot s ; \mathtt{n := 0})\right)\right)
\]
(Note that $\textbf{1}$ represents the natural number one and $s$ represents the cost of the skis.) This program can be represented by the expression
\[ \left (\mathtt{p}\left (\mathtt{a} \odot \textbf{1} + (\mathtt{a} \odot s)\mathtt{b} \right ) \right )^{*} \bar{\mathtt{p}}\] 
\end{example}

\begin{example}
As discussed in \cite{BatzEtAl2022} \emph{probabilistic choice} ``do $e$ with probability $s \in [0,1]$ and $f$ with probability $1-s$'' can be expressed in the language of weighted programs.  Using weighted Kleene algebra (Boolean tests are not necessary), we can express probabilistic choice as \[ (e \odot s) + (f \odot (1-s)) \, . \] A more thorough examination of the relation of weighted Kleene algebra to \emph{probabilistic regular expressions} \cite{RozowskiSilva2023} is left for another occasion.
\end{example}

\section{Language completeness}\label{sec:Lan}

In this section we define a weighted version of guarded strings (Section \ref{sec:Lan-1}) and we prove our first completeness result connecting Kleene $S$-algebras with tests and guarded formal power series over $S$ (Section \ref{sec:Lan-2}).

\subsection{Guarded strings}\label{sec:Lan-1}

Take a finite alphabet $\Phi$ and assume that it is ordered in some arbitrary but fixed way as $\mathtt{p}_1, \ldots, \mathtt{p}_n$. Recall from \cite{Kozen1997} that a \emph{$\Phi$-atom} is a string $a_1 \ldots a_n$ over literals $\Lambda = \Phi \cup \bar{\Phi} = \Phi \cup \{ \bar{\mathtt{p}} \mid \mathtt{p} \in \Phi \}$ where $a_i \in \{ \mathtt{p}_i, \bar{\mathtt{p}}_i \}$. Let $\mathrm{At}(\Phi)$ be the set of all $\Phi$-atoms. For $G \in \mathrm{At}(\Phi)$ and $\mathtt{p} \in \Phi$, we write $G \vDash \mathtt{p}$ if there is $i \leq n$ such that $G(i) = \mathtt{p}$. $\mathrm{At}(\Phi)$ represents all possible assignments of truth values to propositional letters in $\Phi$. A \emph{guarded string over $\Sigma$ and $\Phi$} is a string in $(\mathrm{At}(\Phi) \cdot \Sigma)^{*} \cdot \mathrm{At}(\Phi)$; in other words, it is a string of the form
\[ G_1 \mathtt{a}_1 G_2 \ldots \mathtt{a}_{n-1} G_n \, ,\] where $G_i \in \mathrm{At}(\Phi)$ and $\mathtt{a}_j \in \Sigma$. The set of guarded strings over $\Sigma$ and $\Phi$ will be denoted as $\mathrm{Gs}(\Sigma, \Phi)$. We do not distinguish between non-empty words $a_1 \ldots a_n$ over $(\Sigma \cup \Lambda)$ and the corresponding expressions $a_1 \cdot \ldots \cdot a_n$ (assuming some fixed bracketing). Hence, every guarded string can also be seen as an expression in $\mathrm{Exp}(\Sigma, \Phi, S)$ or even as an expression in $\mathrm{Exp}(\Sigma \cup \Lambda, S)$. 

For a guarded string $\sigma = G_1 \mathtt{a}_1 G_2 \ldots \mathtt{a}_{n-1} G_n$, we define $\mathsf{head}(\sigma) = G_1$, $\mathsf{tail}(\sigma) = G_n$ and $\mathsf{body}(\sigma) = \mathtt{a}_1 G_2 \ldots \mathtt{a}_{n-1} G_n$ (not a guarded string). The \emph{fusion product} is a partial binary operation $\diamond$ on $\mathrm{Gs}(\Sigma, \Phi)$ such that
\[ \sigma \diamond \sigma' = 
\begin{cases}
\mathsf{head}(\sigma)\mathsf{body}(\sigma')& \mathsf{tail}(\sigma) = \mathsf{head}(\sigma');\\ 
\text{undefined} & \mathsf{tail}(\sigma) \neq \mathsf{head}(\sigma').
\end{cases}
\]
(For instance, $G \mathtt{a} H \diamond H \mathtt{b} G = G \mathtt{a} H \mathtt{b} G$, but $G \mathtt{a} H \diamond G \mathtt{a} H$ is undefined if $G \neq H$.) Fusion product is lifted to a total binary operation on subsets of $\mathrm{Gs}(\Sigma, \Phi)$ in the obvious way. 


The free multimonoid of guarded strings $\Sigma_\Phi^{*}$ is the partial monoid with universe $\mathrm{Gs}(\Sigma, \Phi)$, partial multiplication $\diamond$, and $\mathrm{At}(\Phi)$ as the set of unit elements.\footnote{For much more on multimonoids, see \cite{KudryavtsevaMazorchuk2015} and \cite{FahrenbergEtAl2023}.}

\begin{definition}
Let $S$ be a finite semiring. The set of \emph{guarded formal power series} over ($\Sigma, \Phi$ and $S$) is the set $\fps{S}{\Sigma_\Phi^{*}}$ of mappings from $\mathrm{Gs}(\Sigma, \Phi)$ to $S$. The \emph{rational operations} on $\fps{S}{\Sigma_\Phi^{*}}$ are defined point-wise as follows:
\begin{itemize}
\item unit: $1(w) = 1$ if $w \in \mathrm{At}(\Phi)$ and $1(w) = 0$ otherwise;
\item annihilator: $0(w) = 0$ for all $w$;
\item addition: $(r_1 + r_2)(w) = r_1(w) + r_2(w)$;
\item multiplication: $(r_1 \cdot r_2)(w) = \sum \{ r_1(v_1) \cdot r_2(v_2) \mid w = v_1 \diamond v_2 \}$;
\item scalar multiplication for $s \in S$: $(r \odot s)(w) = r(w) \cdot s $;
\item Kleene star: $r^{*}(w) = \sum_{n \in \omega} r^{n}(w)$, where $r^{0} = 1$ and $r^{n+1} = r^{n} \cdot r$. 
\end{itemize}
A \emph{polynomial} in $\fps{S}{\Sigma^{*}_\Phi}$ is any $r \in \fps{S}{\Sigma^{*}_\Phi}$ such that the set of $w \in \mathrm{Gs}(\Sigma, \Phi)$ where $r(w) \neq 0$ is finite. The set $\fps{S^{\mathrm{rat}}}{\Sigma^{*}_\Phi}$ of \emph{rational guarded formal power series} is the least subset of $\fps{S}{\Sigma^{*}_\Phi}$ that contains all polynomials and is closed under the rational operations.
\end{definition}
Note that $\sum_{n \in \omega} r^{n}(w)$ exists since $S$ is assumed to be finite.

It is easily seen that $\fps{S^{\mathrm{rat}}}{\Sigma^{*}_\Phi}$ is an idempotent semiring. If $K \subseteq \mathrm{Gs}(\Sigma, \Phi)$ is finite and $s \in S$, then $s_K$ denotes the formal power series (in fact, a polynomial) such that $s_K(w) = s$ if $w \in K$ and $s_K(w) = 0$ otherwise. Given the definition of fusion product, the subalgebra of  $\fps{S^{\mathrm{rat}}}{\Sigma^{*}_\Phi}$ with the universe $\{ 1_K \mid K \subseteq \mathrm{At}(\Phi) \}$ is a distributive lattice, naturally extended to a Boolean algebra by defining $\overline{(1_K)} = 1_{\mathrm{At}(\Phi) \setminus K}$. In this way, we obtain a Kleene $S$-algebra with tests:

\begin{lemma}\label{lem:Lan-sound}
$\fps{S^{\mathrm{rat}}}{\Sigma_\Phi^{*}}$ is a Kleene $S$-algebra with tests.
\end{lemma}

\begin{definition}
The \emph{standard language interpretation} of $\mathrm{Exp}(\Sigma, \Phi, S)$ is the unique homomorphism $\Gu : \mathrm{Exp}(\Sigma, \Phi, S) \to \fps{S}{\Sigma^{*}_\Phi}$ such that
\[ \Gu (\mathtt{a}) = 1_{\{ G\mathtt{a}H \mid G, H \in \mathrm{At}(\Phi) \}}
\quad\text{and}\quad
\Gu (\mathtt{p}) = 1_{\{ G \mid G \vDash \mathtt{p} \}} \, .
\]
\end{definition}
It is easily seen that, in fact, $\Gu : \mathrm{Exp}(\Sigma, \Phi, S) \to \fps{S^{\mathrm{rat}}}{\Sigma^{*}_\Phi}$.

\subsection{The language completeness result}\label{sec:Lan-2}

This section establishes our first result: if $S$ is a finite copi-semiring, then the equational theory of $\mathsf{KAT}(S)$ coincides with the equational theory of a ``canonical model'' consisting of $S$-weighted guarded formal power series. 

From now on we assume that every expression $e \in \mathrm{Exp}(\Sigma, \Phi, S)$ is in \emph{Boolean normal form}, that is, $\bar{b}$ occurs in $e$ only if $b \in \Phi$. Recall that every expression $e \in \mathrm{Exp}(\Sigma, \Phi, S)$ in Boolean normal form can be seen as an expression in $\mathrm{Exp}(\Sigma \cup \Phi \cup \bar{\Phi}, S)$, and so it makes sense to consider the standard language interpretation $\Ln (e) \in \fps{S}{(\Sigma \cup \Lambda)^{*}}$ of $e$. 

Our proof of the language completeness result for $\mathsf{KAT}(S)$ is very similar to Kozen and Smith's proof of the language completeness result for $\mathsf{KAT}$ \cite{KozenSmith1997}; what gives some novelty to our proof is that we need to take weights into account. The proof strategy is the following one: we identify a set of expressions $F \subseteq \mathrm{Exp}(\Sigma, \Phi, S)$ such that $\Ln (f) = \Gu(f)$ for all $f \in F$, and we show that every $e \in \mathrm{Exp}(\Sigma, \Phi, S)$ is $\mathsf{KAT}(S)$-equivalent to some $f \in F$.


\begin{definition}\label{def:guarded-sum}
A \emph{guarded expression} (over $\Sigma$, $\Phi$ and $S$) is any expression of the form $\mathtt{0} \odot 1$, $ G \odot s$ or $GeH \odot s$ where $G, H \in \mathrm{At}(\Phi)$, $e \in \mathrm{Exp}(\Sigma, \Phi, S)$, and $s \in S$. A \emph{guarded sum} is any expression of the form $\sum_{i < n} e_i$ where each $e_i$ is a guarded expression. (The empty sum is defined to be the guarded expression $\mathtt{0} \odot 1$.) 


We define $\mathsf{head}(G \odot s) = G = \mathsf{head}(GeH \odot s)$, $\mathsf{tail}(H \odot s) = H = \mathsf{tail}(GeH \odot s)$, $\mathsf{body}(H \odot s) = \epsilon$, $\mathsf{body}(GeH \odot s) = eH$, and $\mathsf{weight}(e \odot s) = s$. For guarded expressions $e, f$, we define
\[ e \Gdia f =
	\begin{cases}
	(e\cdot \mathsf{body}(f)) \odot (\mathsf{weight}(e) \cdot \mathsf{weight}(f)) &
			\mathsf{tail}(e) = \mathsf{head}(f);\\
	\mathtt{0} \odot 1 &
			\mathsf{tail}(e) \neq \mathsf{head}(f). 		
	\end{cases}
 \]
\end{definition}

\begin{definition}\label{def:guarded-sum-operations}
We define the following operations on guarded sums:
\begin{itemize}
\item $\Gone = \sum_{ G \in \mathrm{At}(\Phi)} (G \odot 1)$ and $\Gzero = \mathtt{0} \odot 1$;
\item $e \Gplus f = e + f$;
\item $(e_1 + \ldots + e_n) \Gdot (f_1 + \ldots + f_m) = 
		\sum_{i \leq n, j \leq m} (e_i \Gdia f_j)$;
\item $(e_1 \odot s_1 + \ldots + e_n \odot s_n) \G{\odot t} = 
	 (e_1 \odot (s_1 \cdot t) + \ldots + e_n \odot (s_n \cdot t))$;		
\item $(g_1 + \ldots + g_n)^{\Gstar}$ is defined by induction on $n$ as follows:
		\begin{enumerate}
			\item $(\mathtt{0} \odot 1)^{\Gstar} = \Gone = (G \odot s)^{\Gstar}$;
			\item $(GeH \odot s)^{\Gstar} = 
					\begin{cases}
						\Gone + (GeH \odot s) & G\neq H;\\
						\Gone + (G e (G e \odot s)^{*} H \odot s) & G = H;
					\end{cases}$
			\item If $ e = (g_1 + \ldots + g_{n-1})$ for guarded expressions $g_1, \ldots, g_{n-1}$ and $g$ is a guarded expression, then $$(e + g)^{\Gstar} = e^{\Gstar} + \left  ( e^{\Gstar} \Gdot f^{\Gstar} \Gdot g \Gdot e^{\Gstar} \right ) \, ,$$ where $f = \left ( \mathsf{head}(g) \mathsf{body}(g) (e^{\Gstar}) \mathsf{head}(g) \right ) \odot \mathsf{weight}(g)$.		
		\end{enumerate}		
\end{itemize}
\end{definition}
Note that, in the definition of $(e + g)^{\Gstar}$, $f$ is a guarded expression and so $f^{\Gstar}$ is provided by the induction steps 1 or 2. Note also that the set of guarded sums $\mathrm{Gsu}(\Sigma, \Phi, S)$ endowed with the operations defined in Def.~\ref{def:guarded-sum-operations} can be seen as a term algebra of the Kleene $S$-algebra type.

\begin{definition}\label{def:hat}
For all $e \in \mathrm{Exp}$, we define $\widehat{e} \in \mathrm{Gsu}$ by induction on $e$ as follows:
\begin{gather*}
\widehat{\mathtt{a}} = \sum_{G, H \in \mathrm{At}(\Phi)} G \mathtt{a} H \odot 1
\qquad \widehat{\mathtt{p}} = \sum_{G \vDash \mathtt{p}} G \odot 1
\qquad \widehat{\mathtt{\bar{p}}} = \sum_{G \not\vDash \mathtt{p}} G \odot 1
\\
\widehat{\mathtt{1}} = \Gone \qquad
\widehat{\mathtt{0}} = \Gzero \qquad
\widehat{e + f} = \widehat{e} \Gplus \widehat{f} \qquad
\widehat{e \cdot f} = \widehat{e} \Gdot \widehat{f}
\\[2mm]
\widehat{e \odot s} = \widehat{e} \G{\odot s} \qquad\quad
\widehat{e^{*}} = (\widehat{e})^{\Gstar}
\end{gather*}
\end{definition}

\begin{lemma}\label{lem:hat}
For every $e \in \mathrm{Exp}$:
\begin{enumerate}
\item $\mathsf{KAT}(S) \models e \approx \widehat{e}$;\label{lem:hat-1}
\item $\Gu (\widehat{e}) = \Ln (\widehat{e})$.\label{lem:hat-2}
\end{enumerate}
\end{lemma}
\begin{proof}
See Appendix \ref{a:hat}. We note that our assumption that $S$ be integral seems to be required in the proof. 
\end{proof}

\begin{theorem}\label{thm:Lan}
For all $e, f \in \mathrm{Exp}(\Sigma, \Phi, S)$,
\[ \mathsf{KAT}(S) \models e \approx f \iff \fps{S^{\mathrm{rat}}}{\Sigma^{*}_\Phi} \models e \approx f\, .\]
\end{theorem}
\begin{proof}
The implication from left to right follows from Lemma \ref{lem:Lan-sound}. The converse implication is established as follows:
\begin{align*}
\fps{S^{\mathrm{rat}}}{\Sigma^{*}_\Phi} & \models e \approx f\\
\fps{S^{\mathrm{rat}}}{\Sigma^{*}_\Phi} & \models \widehat{e} \approx \widehat{f} 
		&& \text{by Lemma \ref{lem:hat}\eqref{lem:hat-1} and Lemma \ref{lem:Lan-sound}}\\
\Gu (\widehat{e}) & \hspace*{1pt}= \Gu (\widehat{f})
		&& \text{by definition}\\
\Ln (\widehat{e}) & \hspace*{1pt}= \Ln (\widehat{f}) 
		&& \text{by Lemma \ref{lem:hat}\eqref{lem:hat-2}}\\
\mathsf{KA}(S) & \models \widehat{e} \approx \widehat{f} 
		&& \text{by Theorem \ref{thm:EK}, Corollary \ref{coro:EK}}\\
\mathsf{KAT}(S) & \models \widehat{e} \approx \widehat{f}\\ 
\mathsf{KAT}(S) & \models e \approx f 
		&& \text{by Lemma \ref{lem:hat}\eqref{lem:hat-1}}
\end{align*}
\end{proof}

\begin{corollary}\label{coro:Lan}
For all $e, f \in \mathrm{Exp}(\Sigma, \Phi, S)$,
\[ \mathsf{KAT}(S) \models e \approx f \iff \Gu (e) = \Gu (f)\, .\]
\end{corollary}

\section{Relational completeness}\label{sec:Rel}

In this section we define $S$-transition systems (Sect.~\ref{sec:Rel-1}) and we prove our second completeness result connecting Kleene $S$-algebras with tests and $S$-transition systems (Sect.~\ref{sec:Rel-2}). Our proof uses a Cayley-like construction going back to the work of Pratt \cite{Pratt1980} on dynamic algebras; see also the relational completeness result for Kleene algebras with tests \cite{KozenSmith1997}.

\subsection{Transition systems}\label{sec:Rel-1}

\begin{definition}
Let $S$ be a finite semiring. An \emph{$S$-transition system for $\Sigma$ and $\Phi$} is a triple $M = \langle Q, \mathsf{rel}_M, \mathsf{sat}_M \rangle$ where $Q$ is a countable non-empty set and
\[ \mathsf{rel}_M : \Sigma \to S^{Q \times Q} \qquad
\mathsf{sat}_M : \Phi \to S^{Q} \]
\end{definition}
As discussed in Section \ref{sec:semirings}, $M$ can be seen as a function assigning to each $\mathtt{a} \in \Sigma$ a (possibly infinite) $Q \times Q$ matrix $M(\mathtt{a})$ with entries in $S$, and to $\mathtt{p} \in \Phi$ a (possibly infinite) $Q \times Q$ diagonal matrix $M(\mathtt{p})$ with entries in $\{ 0, 1 \}$. $M$ can be extended to a function from  $\mathrm{Exp}(\Sigma, \Phi, S)$ to (possibly infinite) $Q \times Q$-matrices over $S$ using the notions introduced in Section \ref{sec:semirings}:
 \begin{center}
$M(e + f) = M(e) + M(f)$ \quad 
$M(e \cdot f) = M(e) \cdot M(f)$ \quad
$ M(e^{*}) = M(e)^{*}$\\[2mm]
$M(e \odot s) = M(e) \cdot M(s)$ \quad
$M(\mathtt{1}) = 1$ \quad $M(\mathtt{0}) = 0$\quad
$M(\bar{b}) = \overline{M(b)}$  
 \end{center}
where $M(s)$ is the diagonal $Q \times Q$ matrix where the entries in the main diagonal are all $s$; $1$ is the $Q \times Q$ identity matrix and $0$ is the $Q \times Q$ zero matrix; and where $\overline{N}$ for a matrix $N$ with entries in $\{ 0, 1 \}$ is the matrix that results from $N$ by switching all $1$s to $0$s and vice versa. We denote as $\mathrm{Mat}(Q, S)$ the set of all $Q \times Q$ matrices with entries in $S$.

\begin{lemma}\label{lem:Rel}
For all $Q$ and finite $S$, $\mathrm{Mat}(Q, S) \in \mathsf{KAT}(S)$.
\end{lemma}

 
\subsection{The relational completeness result}\label{sec:Rel-2}

\begin{theorem}\label{thm:Rel}
For all $e, f \in \mathrm{Exp}(\Sigma, \Phi, S)$,
\[ \mathsf{KAT}(S) \models e \approx f \iff (\forall M)(M(e) = M(f)) \, .\]
\end{theorem}
\begin{proof}
The implication from left to right follows from Lemma \ref{lem:Rel}. The converse implication is established as follows. Define $\mathsf{cay} : \fps{S^{\mathrm{rat}}}{\Sigma^{*}_{\Phi}} \to \mathrm{Mat}(\mathrm{Gs}(\Sigma, \Phi), S)$ as follows:
\[ \mathsf{cay}(r)_{w, v} = 
\begin{cases}
r(u) & v = w \diamond u\\ 0 & \text{otherwise.}
\end{cases}
\]
The function $\mathsf{cay}$ is injective. Indeed, if there is $w$ such that $r_0(w) \neq r_1(w)$, then take the head of $w$, that is, the unique $H \in \mathrm{At}(\Phi)$ such that $w = H \diamond w$. Note that $\mathsf{cay}(r_0)_{H,w} = r_0(w) \neq r_1(w) = \mathsf{cay}(r_1)_{H, w}$. Define $M = \langle \mathrm{Gs}, \mathsf{rel}_M, \mathsf{sat}_M\rangle$ by $\mathsf{rel}_M(\mathtt{a}) = \mathsf{cay}(\Gu (\mathtt{a}))$ for all $\mathtt{a} \in \Sigma$ and $\mathsf{sat}_M(\mathtt{p}) = \mathsf{cay}(\Gu (\mathtt{p}))$ for all $\mathtt{p} \in \Phi$. Now we claim the following:
\begin{claim}
For all $e$, $\mathsf{cay}(\Gu (e)) = M(e)$.
\end{claim}
The theorem is then immediate: If $\mathsf{KAT}(S) \not\models e \approx f$, then $\Gu (e) \neq \Gu (f)$ by Corollary \ref{coro:Lan}, hence $\mathsf{cay}(\Gu (e)) \neq \mathsf{cay}(\Gu (f))$, and so $M(e) \neq M(f)$. It remains to prove the claim, which we do in Appendix \ref{a:Rel-complet}.
\end{proof}

\section{Conclusion}\label{sec:conclusion}

We extended Ésik and Kuich's \cite{EsikKuich2001} completeness result for finitely weighted Kleene algebras to finitely valued Kleene algebras with tests. This result contributes, we hope, to a better understanding of the properties of frameworks that may be used to formalize reasoning about weighted programs \cite{BatzEtAl2022}.

A number of interesting problems need to be left for the future. First, we would like to prove that the equational theory of $\mathsf{KAT}(S)$ is decidable and establish its complexity. Second, we will work on adapting Kozen and Smith's \cite{KozenSmith1997} argument showing that the quasi-equational theory of $\mathsf{KAT}(S)$ where all assumptions are of the form $e \approx \mathtt{0}$ reduces to the equational theory. Third, it would be interesting to extend Ésik and Kuich's completeness result to Kleene algebras weighted in a more general class of semirings (including, for instance, the {\L}ukasiewicz  semiring) and then to use this result to provide a similar generalization for weighted Kleene algebras with tests. Indeed, it would be interesting to figure out a weighted generalization of other completeness proofs for Kleene algebra, such as \cite{Kappe2023a}. An intriguing problem is also to extend our result to variants of Kleene algebra with tests where tests do not form a Boolean algebra \cite{GomesEtAl2019}.



\subsubsection*{Acknowledgement}
This work was supported by the grant 22-16111S of the Czech Science Foundation. The author is grateful to three anonymous reviewers for their comments.


\appendix
\section{Technical appendix}

\subsection{Proof of Lemma \ref{lem:hat}}\label{a:hat}

It is clear that if $\mathsf{KAT} \models e \approx f$ for $e,f$ without occurrences of $\odot$, then $\mathsf{KAT}(S) \models e \approx f$ for all $S$. We write $e \equiv f$ instead of $\mathsf{KAT}(S) \models e \approx f$. A Kleene $S$-algebra model for $\mathrm{Exp}(\Sigma \cup \Lambda, S)$ is defined as expected (elements of $\Lambda$ are mapped into arbitrary elements of the model). If $\mathsf{KA}(S) \models e \approx f$ for $e, f \in \mathrm{Exp}(\Sigma \cup \Lambda, S)$, then we write $e \equiv_{\mathsf{KA}} f$. Recall that we write $\mathrm{Exp}$ instead of $\mathrm{Exp}(\Sigma, \Phi, S)$ and $\mathrm{Exp}_{\Lambda}$ instead of $\mathrm{Exp}(\Sigma \cup \Lambda, S)$.

\begin{lemma}\label{lem:guardedsums}
The following hold:
\begin{enumerate}
\item $\mathtt{1} \equiv \Gone$ and $\mathtt{0} \equiv \Gzero$;
\item $e \Gplus e' \equiv e + f$ for all guarded sums $e$ and $e'$;
\item $e \G{\odot s} \equiv e \odot s$;
\item $e \Gdot e' \equiv e \cdot e'$ for all guarded sums $e$ and $e'$;
\item $e^{*} \equiv e^{\Gstar}$ for all guarded sums $e$.
\end{enumerate}
\end{lemma}
\begin{proof}
The first three claims are straightforward. The fourth claim follows from the fact that $g_1 \bullet g_2 \equiv g_1 \cdot g_2$ for all guarded expressions $g_1, g_2$.

The fifth claim is established by induction on the number $n$ of guarded expressions in the sum $e = g_1 + \ldots + g_n$. If $n = 1$, then we have three possibilities to consider:
\begin{itemize}
\item $e = \mathtt{0} \odot 1$. Then $e^{*} \equiv \mathtt{0}^{*} \equiv \mathtt{1} \equiv \Gone = (\mathtt{0}\odot 1)^{\Gstar} = e^{\Gstar}$.
\item $e = G \odot s$. Then $e^{*} \equiv (G \odot s)^{*} \equiv \mathtt{1} \equiv \Gone = e^{\Gstar}$. We note that $(G \odot s)^{*} \equiv \mathtt{1}$ holds since $G \leqq \mathtt{1}$ (meaning that $G + \mathtt{1} \equiv \mathtt{1}$) and $s \leq 1$ imply $(G \odot s)^{*} \leqq (\mathtt{1} \odot 1)^{*} \equiv \mathtt{1}^{*} \equiv \mathtt{1}$. The argument does not seem to go through if $s \not\leq 1$.
\item $e = Ge_1H \odot s$. Let us define $f  = Ge_1H$ to simplify notation. Then $e^{*} \equiv \mathtt{1} + (f \odot s)^{*}(f \odot s)$. We have two possibilities, depending on whether $G = H$:
\begin{enumerate}
\item[(a)] If $G \neq H$, then $e^{\Gstar} = \Gone + (f \odot s)$, and it is sufficient to show that $(f \odot s)^{*}(f \odot s) \equiv (f \odot s)$. The $\geqq$ claim follows from $(f \odot s)^{*} \geqq 1$ and the $\leqq$ claim is established by applying the right-fixpoint law of Kleene algebra to the assumption $(f \odot s) + (f \odot s)(f \odot s) \leqq (f \odot s)$. This assumption holds since $(f \odot s)(f \odot s) \leqq (f \odot s)$ follows from the fact that $(f \odot s)(f \odot s) \equiv (ff \odot s) \odot s \equiv \mathtt{0}$ ($ff \equiv \mathtt{0}$ since $G \neq H$).
\item[(b)] If $G = H$, then we reason as follows:
\begin{gather*}
e^{*} \equiv 1 + (f \odot s)(f \odot s)^{*} \equiv
1 + \left ( Ge_1 (H \odot s) \right ) \left  ( Ge_1 (H \odot s) \right )^{*} \equiv\\
1 + Ge_1 \left ( (H \odot s) Ge_1 \right )^{*} (H \odot s) \equiv
1 +  Ge_1 ( HGe_1 \odot s)^{*} (H \odot s) \equiv\\
\Gone + \left  ( Ge_1(Ge_1 \odot s)^{*}H \odot s \right ) \equiv
e^{\Gstar} \, .
\end{gather*}
\end{enumerate}
\end{itemize} 
Now assume that the claim holds of all guarded sums $f = g_1 + \ldots + g_m$ for some $g_1, \ldots, g_m$, and take $e = g_1 + \ldots + g_m + g$ for some $g$. We reason as follows, writing $f$ for $g_1 + \ldots + g_m$: 
\begin{gather*}
e^{*} = (f + g)^{*} \equiv f^{*} (g f^{*})^{*} \equiv 
f^{*} + f^{*} g f^{*} (g f^{*})^{*}\\
\equiv f^{*} + f^{*} g f^{*} \left (\mathsf{head}(g) g f^{*} \right )^{*} \equiv
f^{*} + f^{*} \left  ( g f^{*} \mathsf{head}(g) \right ) g f^{*} \\
\equiv f^{*} + f^{*} \left  ( \mathsf{head}(g) \mathsf{body}(g) f^{*} \mathsf{head}(g)  \odot \mathsf{weight}(g)\right )^{*} g f^{*} \\
\equiv f^{\Gstar} + f^{\Gstar} \left  ( \mathsf{head}(g) \mathsf{body}(g) f^{\Gstar} \mathsf{head}(g)  \odot \mathsf{weight}(g)\right )^{*} g f^{\Gstar}
\equiv (f + g)^{\Gstar} \equiv e^{\Gstar}
\end{gather*}
\end{proof}


\begin{definition}
A guarded expression $e$ is called \emph{proper} iff $\Ln (e) = \Gu (e)$. A guarded sum is called proper iff it is a sum of proper guarded expressions.
\end{definition}

\begin{lemma}\label{lem:proper}
The set of proper guarded sums is closed under $\Gone$, $\Gzero$, $\Gplus$, $\Gdot$, $\G{\oplus s}$ and $\,^{\Gstar}$.
\end{lemma}

\noindent
\textbf{Lemma \ref{lem:hat}}. {\itshape For all $e \in \mathrm{Exp}$:}
\begin{enumerate}
\item $e \equiv \widehat{e}$;
\item $\Ln (\widehat{e}) = \Gu (\widehat{e})$
\end{enumerate}
\begin{proof}
Both claims are established by induction on the complexity of $e$. 

1. The base case for $e \in \Sigma \cup \Lambda$ is easily established using the semimodule axioms for $\odot$, which establish that $(e + f) \odot s \equiv (e \odot s) + (f \odot s)$ and $e \odot 1 \equiv e$,  and $\mathsf{KAT}$ reasoning. The induction step is established using Lemma \ref{lem:guardedsums}. 

2. The base case follows from Definition \ref{def:hat} and the induction step is easily established using Lemma \ref{lem:proper}.
\end{proof}

\subsection{Proof of Theorem \ref{thm:Rel}}\label{a:Rel-complet}
Let  $M = \langle \mathrm{Gs}, \mathsf{rel}_M, \mathsf{sat}_M\rangle$ where $\mathsf{rel}_M(\mathtt{a}) = \mathsf{cay}(\Gu (\mathtt{a}))$ for $\mathtt{a} \in \Sigma$ and $\mathsf{sat}_M(\mathtt{p}) = \mathsf{cay}(\Gu (\mathtt{p}))$ for $\mathtt{p} \in \Phi$. We prove by induction on the structure of $e \in \mathrm{Exp}(\Sigma, \Phi, S)$ that $\mathsf{cay}(\Gu (e)) = M(e)$. We write $[e]$ instead of $\mathsf{cay}(\Gu (e))$.

The base cases for $\mathtt{a} \in \Sigma$ and $\mathtt{p} \in \Phi$ hold by definition of $M$. The case for $\mathtt{0}$ is trivial. The case for $\mathtt{1}$ follows from the definition of $\Gu$. In particular, $\Gu (\mathtt{1}) = 1_{\mathrm{At}}$ is the function that assigns $1$ to atoms and $0$ to all other guarded strings. Hence, $\mathsf{cay}(1_{\mathrm{At}})$ is the identity matrix.

Induction step for $e \odot s$ is established by showing that $[e \odot s]_{w, v} = [e]_{w, v} \cdot s$. It is sufficient to check this in the case that $w = v \diamond u$. Then $[e \odot s]_{w, v} = (\Gu (e) \odot s)(u) = \Gu (e)(u) \cdot s = [e]_{w, v} \cdot s$ by the induction hypothesis.

The case for $e + f$ is routine. The case for $e \cdot f$ is established by showing that $[e \cdot f]_{w,v} = \sum_{u \in \mathrm{Gs}} [e]_{w, u} \cdot [f]_{u, v}$. It is sufficient to check this for $w, v$ such that $v = w \diamond w'$. Then 
 \begin{gather*}
 [e \cdot f]_{w, v} = \Gu (e \cdot f)(w')
 = \sum \{ \Gu (e)(u') \cdot \Gu (f)(v') \mid w' = u' \diamond v'\}\\
 = \sum \{ \Gu (e)(u') \cdot \Gu (f)(v') \mid v = w \diamond (u' \diamond v') \}\\
 = \sum \{ \Gu (e)(u') \cdot \Gu (f)(v') \mid v = u \diamond v' \And u = w \diamond u' \}\\
 = \sum \{ [e]_{w, u} \cdot [f]_{u, v} \mid v = u \diamond v' \And u = w \diamond u' \}\\
 = \sum_{u \in \mathrm{Gs}} [e]_{w, u} \cdot [f]_{u, v}\, .
 \end{gather*}
The case for $e^{*}$ is established by showing that $[e^{*}]_{w, v} = \sum
_{n \in \omega} [e]^{n}_{w, v}$. This is implied by the following claim which is established easily by induction on $n$:

\begin{claim}
For all $n \in \omega$: for all $e \in \mathrm{Exp}$ and all $w, u, v \in \mathrm{Gs}$, if $v = w \diamond u$, then $\Gu (e)^{n}(u) = \sum_{x \in [w,v], |x| = n } M(e)(x)$.
\end{claim}

The case for $\bar{b}$ is established by showing that $\cayG (\bar{b})_{w, v} = \overline{[b]}_{w, v}$. Since $[b]$ is a diagonal matrix, it is sufficient to check the case $w = v$. Then $[\bar{b}]_{w,w} = \Gu (\bar{b})(G)$ for the tail $G$ of $w$, that is, the unique $G \in \mathrm{At}$ such that $w = w \diamond G$. It follows that $\Gu (\bar{b})(G) = \overline{\Gu (b)(G)} = \overline{[b]}_{w,w}$.

\end{document}